\documentclass[journal]{IEEEtran}
 \usepackage{amsthm}
\usepackage{graphicx}
\usepackage{subcaption}
\usepackage{booktabs}
\usepackage{xcolor}

\newtheorem{theorem}{Theorem}
\newtheorem{lemma}{Lemma}

\newtheorem{proposition}{Proposition}

\newtheorem{corollary}{Corollary}[theorem]

\usepackage{cite}

\usepackage{bm}

\usepackage{graphicx}
\usepackage{amsmath}
\usepackage{amssymb}
\usepackage{mathtools}

\usepackage{algorithm}
\usepackage{algorithmic}
\makeatletter
\def\BState{\State\hskip-\ALG@thistlm}
\makeatother

\usepackage{setspace}

\usepackage{amsthm}
\usepackage{xcolor}
\usepackage{lipsum}
\usepackage[margin=0.75in]{geometry}
\newgeometry{top=1in, left=0.75in,right=0.75in,bottom=0.75in}
\def\randE{\ensuremath\mathcal{E}}

\begin{document}
\setlength{\belowcaptionskip}{-10pt}
\setlength{\abovedisplayskip}{2pt}
\setlength{\belowdisplayskip}{2pt}
\setlength{\parskip}{0pt}
%
\title{Mobility-Aware Electric Vehicle Fast Charging Load Models with Geographical Price Variations}

\author{\IEEEauthorblockN{Ahmadreza Moradipari\IEEEauthorrefmark{1},
Nathaniel Tucker\IEEEauthorrefmark{1}, and
Mahnoosh Alizadeh\\\vspace{0.0cm}}

\IEEEauthorblockA{\IEEEauthorrefmark{1}Authors have equal contribution}
}


%


\pagenumbering{gobble}

\maketitle


\begin{abstract}
\textcolor{black}{
We study the traffic patterns as well as the charging patterns of a population of cost-minimizing EV owners traveling and charging within a transportation network equipped with fast charging stations (FCSs). Specifically, we study how the charging network operator (CNO) can influence where EV users charge in order to optimize the utilization of fast charging stations. 
These charging decisions of private EV owners affect aggregate congestion at stations (i.e., waiting time) as well as the aggregate EV charging load across the network. In this work, we capture the resulting equilibrium wait times and electricity load  through a  so-called \textit{traffic and charge assignment problem} (TCAP) in a fast charging station network. Our  formulation allows us to: 1) Study the expected station wait times as well as the probability distribution of aggregate charging load  of EVs in a FCS network in a mobility-aware fashion (an aspect unique to our work), while accounting for heterogeneities in users' travel patterns, energy demands, and geographically variant electricity prices. 2) Analytically characterize the special threshold-based structure that determines how EV owners choose where to charge their vehicle at equilibrium, in response to the FCS's charging price structure, users' energy demands, and users' mobility patterns. 
3) Provide a convex optimization problem formulation to identify the network's unique equilibrium. Furthermore, we illustrate how to induce a socially optimal charging behavior   by deriving the socially optimal plug-in fees and electricity prices at the charging stations.
}

\end{abstract}


%
\IEEEpeerreviewmaketitle

\makeatletter
\def\blfootnote{\xdef\@thefnmark{}\@footnotetext}
\makeatother

\blfootnote{
\indent
This work was supported in part by NSF grant \#1847096 and in part by UCSB's IEE Excellence in Research Graduate Fellowship. The authors are with the Department of Electrical and Computer Engineering, 
University of California, Santa Barbara,
California, 93106, USA (email: nathaniel\_tucker@ucsb.edu).}


\section{Introduction}\label{section: Intro}
As of January 2020, the United States has over one million electric vehicles (EVs) on the road and nearly 25,000 EV charging stations \cite{Intro_one_million, Intro_DOE}. It is important to note that both the EVs and stations are not evenly distributed throughout the country. There are certain regions where EV adoption rates are much higher than average yet the number of charging stations is lacking. 
Additionally, the recharging process of an EV is significantly slower than the refueling process for an internal combustion engine vehicle (ICEV), meaning that EVs occupy chargers for long periods of time. As a result of the limited infrastructure and long charging times, EV owners in populated urban areas are experiencing high levels of congestion at public charging stations during peak usage hours\cite{intro_tesla, Intro_tesla2}. Furthermore,  temporal load shifting is already not possible at fast charging stations due to the requirement to recharge arriving EVs immediately. Hence, unless users are guided to charge at the {\it right stations}, one cannot  control the effect of fast charging station loads on the power grid.

To address the aforementioned issue of charging station congestion, a charging network operator (CNO) has two main options: 1) expand the current charging station network with more chargers and locations, or 2) optimize the usage of the current charging network infrastructure via heterogeneous charging fees that guide EV users away from congested stations or stations with high electricity prices. Regarding the first option, much work has been done for determining the optimal capacities and locations for new charging stations given population data and mobility patterns \cite{8950037,8999636,8984288,7439861,8932405,lam2014electric,luo2015placement,zhang2016pev,abdalrahman2019pev,bayram2015capacity,zhang2019expanding}. However, enlarging a charging station network is an expensive operation and the charging stations added to satisfy demand during peak hours might be unused during off-peak hours. As such, before investing in additional charging infrastructure, CNOs should consider utilizing pricing strategies to optimize the usage of their current charging network. 

In order to optimize the utilization of fast charging stations, the CNO must account for the numerous factors that influence where EV owners decide to charge their vehicles including mobility patterns, plug-in fees at stations, charging prices, and station wait times \cite{8824110,moradipari2018pricing}. For example, each EV owner would face a  trade-off between charging stations with potentially higher charging costs but less waiting time, 
versus those stations with lower charging costs but longer waiting time. 
Clearly, these charging decisions of private EV owners affect aggregate congestion at stations (i.e., waiting time) as well as the aggregate EV charging load across the network. In this work, we capture the resulting equilibrium wait times and electricity load  through a  so-called \textit{traffic and charge assignment problem} (TCAP) in a fast charging station (FCS) network. \textcolor{black}{The main contributions of our modeling framework are the following:
\begin{itemize}
    \item We study the expected station wait times as well as the probability distribution of aggregate charging load  of EVs in a FCS network in a mobility-aware fashion (an aspect unique to our work), while accounting for heterogeneities in users' travel patterns, energy demands, and geographically variant electricity prices. 
    \item We analytically characterize the special threshold-based structure that determines how EV owners choose where to charge their vehicle at equilibrium, in response to the FCS's charging price structure, users' energy demands, and users' mobility patterns. Unlike previous work, we do not need to discretize the charge requests of EVs in order to determine the equilibrium charging behavior. 
    \item We provide a convex optimization problem formulation to identify the network's unique equilibrium. Furthermore, we illustrate how to induce a socially optimal charging behavior   by deriving the socially optimal plug-in fees and electricity prices at the charging stations.
\end{itemize}
}
\textcolor{black}{
The work presented in this manuscript advances the state-of-the-art in traffic and charging network analysis in several main aspects. First, our work directly analyzes the effects of users' charge demands and heterogeneous charging prices on both the traffic flow and the charging load distribution at equilibrium in a charging network. Second, our framework does not make use of energy requests that come from a discrete set, rather, we are able to account for EV users with realistic energy requests from continuous distributions without sacrificing tractability. Third, we advance the state-of-the-art in modeling charging trips by alleviating the need to use excessive number of imaginary arcs/nodes in the network graph to account for different charge durations at a station, thus maintaining tractability for larger networks.
}

\textcolor{black}{
We would also like to briefly comment on the practicality of our framework. Our framework provides insight into the emergent equilibrium of traffic and charging networks and gives information about the aggregate behavior of an EV population. This aggregate behavior model can then be used for designing plug-in fees to mitigate congestion in real networks, to forecast charging load at each charging station, or to plan where to expand a charging network. Additionally, our framework complements recent literature regarding the coupling between EV transportation networks and the power grid; therefore, our work could be used to forecast loads directly for the power grid or aid in planning the addition of new power lines or storage capacity.
}

\textit{Related Works:} Several works have emerged focusing on minimizing users' waiting times at charging station queues \cite{qin2011charging, xu2017dynamic}, incentivizing users away from congested stations \cite{bayram2013decentralized, bi2019distributing}, and limiting charging durations \cite{fan2015operation} to ensure other EVs have a chance to charge. Many works focusing on charging station utilization make use of pricing mechanisms to influence the users within the system for various objectives such as profit maximization of individual charging stations \cite{yuan2015competitive}, minimizing fleet costs\cite{zhou2015optimal,moradipari2019mobility}, or exploiting energy storage capabilities \cite{luo2017stochastic}. However, none of the aforementioned papers analyze the effects of users' charge demands and heterogeneous charging prices on both the traffic flow and charging load distribution at equilibrium in a charging network.

In addition to managing congestion at charging stations, pricing strategies can also help manage the effects of EV charging load on the power grid. For example, papers \cite{zhang2015impact, donadee2014stochastic, lam2015capacity, vaya2015self} focus on using charging stations to provide ancillary support to benefit the local distribution grid. Furthermore, due to the inherent interdependencies of charging stations and the power grid, many researchers have studied the coupled infrastructure systems to find operational strategies that benefit both systems \cite{8932405,7121000,8635950,tucker2019online,alizadeh2016optimal, wei2017network, he2016sustainability, wei2016optimal}.  In studying how users traveling on a transportation network make charging decisions, previous works have made approximations for model tractability. For example, the authors of \cite{alizadeh2016optimal, wei2017network} assume that energy requests come from discrete sets and make use of virtual charging links in their network graphs to account for a limited number of charging decisions. Similarly, the authors of \cite{he2016sustainability, wei2016optimal} assume that all users on the same route will receive the same amount of energy. Our work addresses these challenges by accounting for EV users with energy requests from continuous distributions without sacrificing tractability. Our improved TCAP   supplements the coupled traffic/power system literature and can be directly integrated into frameworks such as \cite{alizadeh2016optimal}.  

The papers most similar to ours study traffic and charging patterns within networks of charging stations at equilibrium; however, the objectives, the models, and the problem settings are different. For example, \cite{jiang2012path} and \cite{wang2016path} focus on traffic assignment but do not focus on charging or the effects of heterogeneous prices, \cite{jiang2014computing} does not focus on pricing nor the distribution of charging load throughout the network, \cite{he2014network} assumes homogeneous charging prices across all stations and is focused on range concerns, and \cite{manshadi2017wireless} assumes there are wireless chargers able to charge the EVs while driving.


The remainder of the manuscript is organized as follows: Section \ref{section: system model} presents the system model and FCS network characteristics, Section \ref{sec: individual user problem} presents the individual EV owner's decision problem, Section \ref{section: network equilibrium model} presents the energy threshold-based structure of the emergent equilibrium in the FCS network, Section \ref{sec.so} discusses the plug-in fees and charging costs that induce the socially optimal (SO) equilibrium, and Section \ref{sec.numerical} presents a numerical example.

\section{System Model}
\label{section: system model}

\subsection{Transportation Network and EV User Characteristics}

In the following, we model the transportation network as a directed graph $\mathcal{G}=(\mathcal{V},\mathcal{A})$, where $\mathcal{V}$ is the set of nodes (locations) and $\mathcal{A}$ is the set of arcs (roads). Furthermore, we assume that there are public fast charging stations located at a subset of the nodes $\mathcal{J}\subset\mathcal{V}$. 

Now, we need to model the effects on station wait times and charging load due to selfish EV owners who want to charge their vehicles. We assume that each EV user wants to travel from an origin node to a destination node that forms a so-called origin-destination (o-d) pair $(o,d)\in\mathcal{O}$. Furthermore, each EV user would like to charge their EV $\epsilon$ units of energy at some fast charging station that they visit during their trip. We wish to emphasize that we are only concerned with \textit{charging trips} in this work, i.e., we do not need to consider trips during which the EV owner does not want to charge their vehicle. 

\textbf{Modeling Assumptions:} Let us now state the assumptions we make.  First, we make two important but standard modeling assumptions: 1) We assume that each EV user contributes an infinitesimal amount of traffic and charge demand compared to the entire EV population as a whole. This allows us to examine the FCS network from a macroscopic level via a network flow model. 2) We assume that the population of EV owners is homogeneous in their value of time (e.g., $\$/\text{minute}$) and denote the inverse of this ratio as $\alpha$ (e.g., $\text{minutes}/\$$).  
We note that these assumptions are common in the study of transportation and charging station networks \cite{alizadeh2016optimal, wei2017network, he2016sustainability, wei2016optimal}.

Regarding the energy request $\epsilon$ of each EV user, we make the following two assumptions: 1) We assume that each EV user has an inelastic energy requirement, meaning that their energy demand is independent of charging station prices. 2) We assume that each EV user's energy requirement is independent of the charging station $j\in\mathcal{J}$ that the user chooses.    We note that  these assumptions are not critical and can be removed at the cost of modeling simplicity. The elasticity of energy demand as a function of charging station prices can be included in our model through the use of elastic demand functions. The second assumption may also be addressed by adding distance dependent deterministic shifts to the energy requests of individual users depending on the charging station they visit. However, we do not believe this is a good modeling choice for the purposes of price design for congestion management and characterization of the distribution of charge requests within public charging station networks (which are commonly concerned with aggregate statistics). Consider the following example regarding an EV owner's average morning trip to work and their decision of which fast charging station to stop at.  If the user travels farther and expends more energy during the charging trip so he can charge his EV at a station closer to his workplace, he will then have more energy left in his battery at the end of the trip. As such, we believe that, at the macro scale, accounting for the energy expended en route during the charging trip is an unnecessary distortion rather than an improvement when considering a fast charging network. However, this can be easily modeled if needed, particularly for long intercity trips. \textcolor{black}{Furthermore, we note that in this work for the purpose of brevity, we only consider one class of EV users. However, our model can readily handle multiple classes of EV users, and this can help us account for range limitations. Namely, each class of EVs would correspond to an initial energy amount and the EVs in each group would be limited to certain stations near their origin. While other groups of EVs with larger initial energy amounts would not be limited in which charging stations they can reach.
}

\subsection{Charging Network Operator (CNO)}
\label{subsection: CNO}

The Charging Network Operator (CNO) controls the operations of the entire charging network. Namely, the CNO sets the plug-in fees and charging prices in order to manage congestion and electricity usage at each FCS. 

Although the characteristics of each individual EV user are private and unknown by the CNO,   the CNO has some degree of knowledge of the mobility patterns and energy requests of the EV population as a whole based on statistics collected in the transportation network. Specifically, to facilitate a congestion mitigation strategy,  the CNO has access to the following information:
\begin{itemize}
    \item Average mobility patterns of the EV population: Let us denote the mean  rate of charging trips between o-d pair $(o,d)$ as $q_{od}$, which is known by the CNO, and is considered inelastic. The results presented in this paper are applicable to 1) the case where the trips occur randomly according to a Poisson process with mean $q_{od}$ and 2) the case where the arrival process is deterministic and constant with rate $q_{od}$. This is because we assume that even if the travel demand is a random process with a constant mean, the users do not observe the current realization of others' charging trips or energy demands; hence, they solely make their path decisions based on expected wait times without arrival information (which would be constant at equilibrium). Considering random arrivals solely helps us with connecting our expected congestion costs to stochastic queuing models.
    \item Distribution of energy requests: We assume that the energy request of each EV user with the o-d pair $(o,d)$ is an i.i.d. random variable $\mathcal{E}$ that is distributed according to a general distribution $g^{(o,d)}_{\mathcal{E}}(\epsilon)$. Furthermore, we denote the minimum and maximum possible energy requests as $\epsilon_{\text{min}}$ and $\epsilon_{\text{max}}$. Without loss of generality, and purely for brevity of notation, we assume that the energy request distributions do not vary between different o-d pairs and the common distribution is denoted as $g^{}_{\mathcal{E}}(\epsilon)$ which is bounded below away from zero in its compact support $[\epsilon_{\text{min}},\epsilon_{\text{max}}]$.
\end{itemize}

\subsection{Charging stations}
As mentioned previously, public charging stations are located at a limited number of nodes $j\in\mathcal{J}$ across the network. Without loss of generality, we assume that all stations can provide similar charging rates of $1/\gamma$ units of energy per time epoch. However, we note that the results in this work can be readily extended to handle heterogeneous charging rates across the  stations in the network\footnote{Mathematically, the effect of different charging rates on the user equilibrium is captured by adding the term $\alpha^{-1}\gamma_j$ to each station's price of electricity. Hence, the additional time spent charging an EV acts as a penalty for stations with slower charge rates, which increases in effect for EVs with larger energy demand.}. 

While traveling and charging throughout the network, EV users can experience congestion at the charging stations. This is primarily due to the fact that there are a limited number of EV chargers at each station which limits the number of EVs that can be simultaneously plugged in. This is worsened when EV owners partake in various activities away from their charging EV which can result in longer than necessary sojourn times that prohibit other EVs from taking their place. In the following, we denote the average number of EV charging requests arriving at station $j$ as $\lambda_j$. Furthermore, we assume that an EV user experiences an average wait time $T_j(\lambda_j)$ associated with finding an available EV charger at station $j$. We assume that $T_j(\lambda_j)$ is a strictly increasing, convex and continuously differentiable function of $\lambda_j$. We provide a discussion on the use of various queuing based models for $T_j(\lambda_j)$ in Section \ref{subsection: queues} in the Appendix. Furthermore, we assume that EV users have access to information about the average congestion levels at each charging station when making their traveling decisions\footnote{We note that this is a less restrictive version of the assumption made in papers that consider user equilibria with stochastic travel times \cite{WATLING20061539}, where it is assumed that customers have knowledge of the joint densities of road travel times.}. For example, this could be achieved via sample averaging that is reported to the CNO's mobile application which the EV users have access to. Making use of the information regarding charging station congestion levels, each EV user's average traveling cost becomes a deterministic function of their energy demand. Specifically, each EV user chooses their minimum average cost path that allows them to travel to their destination and charge en route, which we discuss further in Section \ref{sec: individual user problem}. 

\section{The Individual EV User's Decision Problem}
\label{sec: individual user problem}

Each EV user  is looking for a path from their origin to destination that also stops at a station and allows them to charge their EV en route. Specifically, for an EV user with o-d pair $(o,d)$, an acyclic path $p\in\mathcal{P}$ on the transportation network $\mathcal{G}$ is a feasible path if it connects nodes $o$ and $d$ via a number of consecutive road arcs and enters \textit{a single charging station} $j\in\mathcal{J}$. We define $\mathcal{P}$ as the set of all paths for all o-d pairs and $\mathcal{P}_{od}\subset\mathcal{P}$ as the set of feasible paths for o-d pair $(o,d)$, with $|\mathcal{P}_{od}|$ denoted as $K_{od}$.

\subsection{EV User Costs}

Each EV user's objective is to select the path $p\in\mathcal{P}_{od}$ that minimizes their average trip cost, which consists of travel time costs, waiting costs due to congestion at charging stations, and monetary costs due to charging fees collected by the CNO. We discuss these cost components next.

1) Latency Costs on Road Arcs: To preserve simplicity of exposition, we assume that each road arc $a\in\mathcal{A}$ has constant latency $t_a >0$ that is independent of the mean flow rate $x_a$ of users on that road arc. This is a reasonable assumption provided that the flow of EVs is small relative to the flow of other vehicles on the network, which is indeed true in transportation systems at the time of this writing \cite{Low_penetration}. We note that the results of this paper can be readily extended to handle separable cost functions for road arcs.

2) Latency Costs at Charging Stations: Each user's average sojourn time at charging station $j$ is dependent on the charging station's congestion as well as the user's energy request $\epsilon$. Specifically, a user's average sojourn time at charging station $j$ can be calculated as 
\begin{align}
    \text{average sojourn time} = \gamma \epsilon +T_{j}( \lambda_{j})
\end{align}
where the first term corresponds to the time spent charging and the second term corresponds to the time spent searching for an available charger due to congestion.

Combining the latency costs on road arcs and at charging stations, we can calculate the total latency cost of any path $p_i\in\mathcal{P}$. Specifically, given the congestion at all charging stations within the network, $\lambda=[\lambda_{j}]_{j\in\mathcal{J}}$, the latency on path $p_i$ can be calculated as follows:
\begin{equation}\label{eq:total latency}
l_{p_i}(\lambda) = \sum_{a \in A}  \delta_{ai}t_{a} + \sum_{j \in J}  \delta_{ji}(\gamma \epsilon +T_{j}( \lambda_{j})),
\end{equation}
where $\delta_{ai}=1$ if path $p_i$ includes road arc $a$ and $\delta_{ai}=0$ otherwise. Similarly, $\delta_{ji}=1$ if path $p_i$ enters charging station $j$ and $\delta_{ji}=0$ otherwise.

3) Cost of Charging an EV: Each time an EV user plugs in their EV at charging station $j$, we assume they pay a one time plug-in fee $\tau_j$. Furthermore, they must pay an electricity price $v_j$ for each unit of energy that is transferred to their EV. As such, the payment for an EV user with energy demand $\epsilon$ at charging station $j$ is calculated as
\begin{equation}\label{payment1}
\text{payment at charging station $j$} = \upsilon_{j} \epsilon + \tau_{j}.
\end{equation}
In subsequent sections, we provide more specifics on the objectives of the CNO and to select $\tau_j$ and $v_j$.

\subsection{The EV Users' Path Selection Problem}
Regarding the behavior of the EV user population, our key assumption is that each individual $i$ selects the feasible path $p_i\in\mathcal{P}_{od}$ with the smallest average total cost. Specifically, an EV user with o-d pair $(o,d)$ and energy demand $\epsilon$ would solve for the path $p_i\in\mathcal{P}_{od}$ with the smallest average total cost:
\begin{align} \label{eq:costfunction2}
C_{od}^{p_i}(\epsilon) = & \left[\sum_{a \in A}  \delta_{ai}t_{a} + \sum_{j \in J}  \delta_{ji}(\gamma \epsilon +T_{j}(\lambda_{j}))\right] \nonumber\\
& + \alpha \sum_{j \in J}\delta_{ji}(\tau_{j} + \upsilon_{j} \epsilon).
\end{align}
The last term in \eqref{eq:costfunction2} incorporates the monetary costs due to charging via the population's time-value-of-money parameter $\alpha$. We note that this path choice model is the underlying framework for the aggregate analysis of how individual users' decisions affect the overall station queues and charging load in a Wardropian network equilibrium \cite{correa2011wardrop}.

\section{Network Equilibrium Model}
\label{section: network equilibrium model}
In the following, we discuss the equilibrium that emerges in the charging station network with selfish EV users solving for their optimal charging locations. 
In the following, we will show that the dependence of each individual's cost on their energy demand $\epsilon$ ultimately determines the structure of the equilibrium. Specifically, at equilibrium, a user with a large energy request would select a path with higher total latency to exploit potentially cheaper charging prices. On the other hand, a user with a small energy request would select a path with less total latency but potentially higher charging prices. Due to this behavior, the feasible paths connecting an o-d pair receive unequal portions of the arriving EV users dependent on their energy requests. Moreover, the distribution of energy requests at a given charging station at equilibrium will not match the the EV population's energy demand distribution $g_{\mathcal{E}}(\epsilon)$. Instead, we will see that for each o-d pair, the span of EV users' energy request values $[\epsilon_{\text{min}},\epsilon_{\text{max}}]$ can be partitioned into intervals such that all the EV users with energy demands within a unique interval will be assigned to the same unique path. We will discuss this structure in detail in Section \ref{subsection: threshold wardrop}.

\subsection{Network Flows}

In the following, we denote $f^i_{od}$ as the average flow rate of EV users of o-d pair $(o,d)$ traveling on the $i$-th path $p_i\in\mathcal{P}_{od}$, such that the total flow for the o-d pair is $q_{od}=\sum_{i=1}^{K_{od}}f^i_{od}$. Furthermore, we define $f = [f_{od}^{i}]_{\forall (o,d) \in \mathcal O, i = 1,\dots,K_{od}}$ as the column vector of flows for all combinations of EV user types and paths. According to the conservation of flow constraints, the average arrival rates of users at charging stations and road arc flow rates can be calculated. Namely, the average arrival rate of EV users at any charging station $j\in\mathcal{J}$ satisfies the following equation:
\begin{equation}\label{eq:flow1}
\lambda_{j} = \sum_{(o,d) \in \mathcal O} \sum_{i = 1}^{K_{od}}\delta_{ji}f_{od}^{i}.
\end{equation}
Similarly, the mean flow rate of EV users on any road arc $a\in\mathcal{A}$ satisfies the following equation:
\begin{equation}\label{eq:flow2}
x_{a} = \sum_{(o,d) \in \mathcal O}\sum_{i = 1}^{K_{od}} \delta_{ai}f_{od}^{i}.
\end{equation}

\subsection{A Threshold-Based Wardropian Equilibrium Structure}
\label{subsection: threshold wardrop}

Following the classical definition of Wardrop's first principle (or the extension proposed by Holden \cite{holden1989wardrop}), we know that, at the state of equilibrium, for each OD pair, no trip-maker can decrease their experienced  trip cost with respect to their own energy requirement $\epsilon$  by unilaterally changing paths. Here we leverage this fact to characterize the structure of the FCS network's queue and charging load equilibrium. 

In the following, we denote the electricity price at the charging station on feasible path $p_i\in\mathcal{P}_{od}$ as $\theta_{od}^i$:
\begin{align}\label{pric.each.path}
    \theta_{od}^{i} = \sum_{j \in J} \delta_{ji}\upsilon_j.
\end{align}
Furthermore, we assume that   the paths in the feasible set $\mathcal{P}_{od}$  are indexed such that their electricity prices $\theta_{od}^{i}$  are in a decreasing order, i.e., for all paths $p_i$ and $p_{i+1}$ in $\mathcal{P}_{od}$, we have $\theta_{od}^{i} \geq \theta_{od}^{i+1}$.

\begin{lemma}
\label{lemma: firstprinc}
If a path flow pattern $f$ satisfies Wardrop's first principle, for each two  alternative  paths $p_k$ and $p_i$ connecting one o-d pair $(o,d)$ such that $k\geq i$, the energy requirement $\epsilon$ of any user (if any at all) traveling on path $p_i$  is less than or equal to the energy requirement $\epsilon'$ of any user (if any at all) traveling on path $p_k$.
\end{lemma}
\noindent\textit{Proof.} The proof is in the Appendix.

The observation made in Lemma \ref{lemma: firstprinc} matches the intuition discussed in beginning of Section \ref{section: network equilibrium model}. Namely, EV users with the largest energy requests will select paths with lowest electricity costs. Given this observation, the next theorem characterizes the structure of flows at equilibrium as a function of \textit{energy thresholds} for each o-d pair (see Fig. \ref{fig:thr}). We note that this energy threshold-based structure will later be useful in finding the equilibrium flows as the solution of a nonlinear program.
\begin{figure}
\centering
\includegraphics[width=0.75\columnwidth]{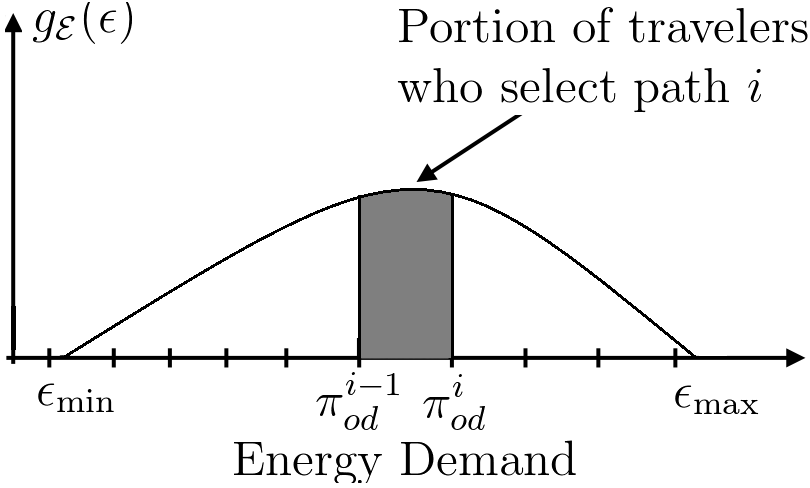} 
\caption{Example distribution of energy demand for EV users. We note that \textit{only} the EV users with energy demands in the shaded region (i.e., the interval $[\pi_{od}^{i-1},\pi_{od}^{i}]$) will select the unique path $p_i$ for the o-d pair $(o,d)$.}
\label{fig:thr}
\end{figure}

\begin{theorem}
\label{theorem: energy thresholds}
The equilibrium traffic pattern $ f_{od} = [f_{od}^{i}]_{i = 1,\dots,K_{od}}$ for each o-d pair $(o,d)$  is  characterized by a vector of energy thresholds  $\boldsymbol{\pi}_{od} = (\pi^0_{od},\pi^1_{od},\ldots,\pi^{K_{od}-1}_{od}, \pi^{K_{od}}_{od})$, where $\pi^0_{od} = \epsilon_{\min}$ and $\pi^{K_{od}}_{od} =  \epsilon_{\max}$. All customers with $\pi^{i-1}_{od} \leq \epsilon \leq \pi^{i}_{od}$ will choose path $p_i$. 

\noindent Moreover, for any used paths $p_i$ and $p_{i+1}$, 
\begin{align}\label{thresholds}
    \pi^{i}_{od} = \frac{\psi_{od}^{i+1} - \psi_{od}^{i}}{\alpha (\theta_{od}^{i} - \theta_{od}^{i+1})},
\end{align}where for path $p_i$ in pair (o,d):
        \begin{align}
            &C_{od}^{p_i}(\epsilon) = \psi_{od}^{i}  + \alpha \theta_{od}^{i}\epsilon
        \end{align}
        and
        \begin{align}   
            & \psi_{od}^{i} = \sum_{a \in A}  \delta_{ai}t_{a} + \sum_{j \in J}  \delta_{ji}(\gamma \epsilon +T_{j}(\lambda_{j})) 
           + \alpha \sum_{j \in J}\delta_{ji}\tau_{j}.
        \end{align}
Furthermore, the vector $\boldsymbol{\pi}_{od}$ characterizes the flow on  each path $p_i$ in $(o,d)$ by: 
\begin{equation} \label{thresholdtoflow}f_{od}^{i} =  q_{od} \left(\int_{\pi_{od}^{i-1}}^{\pi_{od}^{i}} g_{\mathcal{E}}(x) dx\right).\end{equation} 
\end{theorem}
\noindent\textit{Proof.} The proof is in the Appendix.

In summary, Theorem \ref{theorem: energy thresholds} informs the CNO how the population of EV users will select paths based on their energy demands. Namely, Theorem \ref{theorem: energy thresholds} states that all travelers with o-d pair $(o,d)$ can be partitioned based on their energy demand and then all travelers within a partition will select the same unique path between $o$ and $d$ making use of the same charging station. Figure \ref{fig:thr} gives a graphical representation of this threshold structure.  This type of threshold-based  equilibrium structure is similar to the type of user equilibrium (UE) flows if one considers traffic patterns of non-EV drivers with continuous value of time distributions \cite{leurent}.

\begin{corollary}
\label{corollary: equal cost}
At equilibrium, for any utilized paths $p_i$ and $p_{i+1}$,  we must have:
\begin{align} 
\label{thr.UE1} 
C_{od}^{p_i}(\pi^{i}_{od}) = ~ C_{od}^{p_{i+1}}(\pi^{i}_{od}).
\end{align}
\end{corollary}

\noindent Corollary \ref{corollary: equal cost} states that if an EV user's energy request falls on the threshold between two partitions, then the travel cost is equal between the two corresponding paths.

\subsection{The User Equilibrium Optimization Problem}

In this section, we develop a nonlinear minimization problem whose solution(s) satisfy the user equilibrium conditions from Section \ref{subsection: threshold wardrop}.  In the following, we make use of the notation listed in Table \ref{tab.not}.
\begin{table}[]
\begin{center}
\resizebox{0.70\columnwidth}{!}{%
\begin{tabular}{r c p{0.55\columnwidth} }
\toprule

$G(\epsilon) $ & $ \triangleq$ & $\int_{0}^{\epsilon}g_{\mathcal{E}}(x) dx$, the cumulative density function for energy demands $\epsilon$ \\

$G^{-1}(t)$ & $\triangleq$ & $ \text{SUP}\{ \epsilon; G(\epsilon) < t\}$, the inverse function for the energy demand CDF, $G(\epsilon) $\\

$E(x)$ & $\triangleq$ & $\int_{0}^{x} G^{-1}(t) dt$\\


$Q_{od}^{i}$ & $\triangleq$ &  ${ \sum_{j=1}^{i} f_{od}^{j}}$, { sum of flows on paths $1,\dots,i$ (i.e., sum of flows on the paths whose electricity prices are greater than or equal to $\theta_{od}^i$)}  \\  

\bottomrule
\end{tabular}}
\end{center}
\caption{Additional notation.}
\label{tab.not}
\end{table}

A major challenge we encounter in characterizing the equilibria is due to the nonlinear energy partition to flow mapping in \eqref{thresholdtoflow}. To overcome this obstacle, we define the value $\frac{Q^i_{od}}{q_{od}}$ as the portion of the population traveling between the o-d pair $(o,d)$ on the paths $p_1$ to $p_i$. This new variable allows us to write the energy-thresholds $\pi^i_{od}$ from Theorem \ref{theorem: energy thresholds} in the form $\pi^i_{od}=G^{-1}(\frac{Q^i_{od}}{q_{od}})$. Given this relationship, we can now formulate a nonlinear minimization problem whose global solution satisfies the stated user equilibrium conditions in Theorem \ref{theorem: energy thresholds}.

\begin{theorem} 
\label{thr.UE}
The traffic pattern $f$ is a user equilibrium traffic and charge pattern in the charging station network if and only if it solves the following nonlinear optimization problem:
\begin{subequations} 
\label{eq:UE.problem}
	\begin{align}
	& \underset{f}{\textrm{minimize}}
	&& \sum_{a \in A} x_a t_a + \alpha \sum_{j \in J} \lambda_j \tau_j + \sum_{j \in J} \int_{0}^{\lambda_j} T_{j}(x) dx ~ + \nonumber \\ 
	& && \alpha\!\!\! \sum_{(o,d) \in \mathcal O}\!\!q_{od} \! \left( \sum_{i = 1}^{K_{od}} \theta_{od}^{i}  \left[ E(\frac{Q_{od}^{i}}{q_{od}}) - E(\frac{Q_{od}^{i-1}}{q_{od}}) \right] \right) \label{UEopt}\\
		& \text{subject to:} \nonumber\\
	& \forall a \in A:
	& & x_a = \sum_{(o,d) \in \mathcal O} \sum_{i = 1}^{K_{od}}  \delta_{a i} f_{od} ^{i}, \label{eq:UE.cons1}\\
	& \forall j \in J:
	& & \lambda_j = \sum_{(o,d) \in \mathcal O} \sum_{i = 1}^{K_{od}}  \delta_{j i}f_{od} ^{i}, \label{eq:UE.cons2}\\
	&  \forall  (o,d)   :
	& &  q_{od} = \sum_{i = 1}^{K_{od}}  f_{od} ^{i} \label{eq:UE.cons3},\\
	& \forall (o,d), \ \forall i:
	& &  Q_{od}^{i} = \sum_{k \leq i} f_{od}^{k
	} \label{eq:UE.cons4},\\
	& \forall (o,d),  \ \forall i:
	& &  f_{od}^{i} \geq 0 \label{eq:UE.cons5}.
	\end{align}
	\end{subequations}
\end{theorem}

\noindent\textit{Proof.} The proof is in the Appendix.

In Theorem \ref{thr.UE}'s proposed traffic and charge assignment problem formulation, the objective function \eqref{UEopt} accounts for the EV population's collective travel time on road arcs (first term), the EV population's total cost from plug-in fees (second term), the collective congestion experienced at the charging stations (third term), and the total cost due to purchasing energy at the charging stations (fourth term). The constraints are described as follows: \eqref{eq:UE.cons1} calculates the EV flow on each road arc, \eqref{eq:UE.cons2} calculates the EV arrival rate at each charging station, \eqref{eq:UE.cons3} ensures that all the EV flow for an o-d pair is accounted for in the path flows of that o-d pair, \eqref{eq:UE.cons4} calculates the percentiles of EV flow traveling on each path for a given o-d pair, and lastly \eqref{eq:UE.cons5} ensures all EV flows are non-negative. 

Furthermore, in Theorem \ref{uniq.th.eq.} we are able to ensure uniqueness of the equilibrium flows as well as the energy-thresholds for partitioning the EV users, under certain assumptions.

\begin{theorem}
\label{uniq.th.eq.}
If charging station wait time functions $T_j(\lambda_j)$ are  univariate functions of $\lambda_j$, the equilibrium path flows, $[f^i_{od}]_{\forall (o,d)\in\mathcal{O},i=1,\dots,K_{od}}$, as well as the  thresholds, $[G^{-1}(\frac{Q^i_{od}}{q_{od}})]_{\forall (o,d)\in\mathcal{O},i=1,\dots,K_{od}}$, are unique.  
\end{theorem}
\noindent\textit{Proof.} The proof is in the Appendix.

The reader should note that the assumption put forward under Theorem \ref{uniq.th.eq.} does not hold if the wait time to plug in depends on the thresholds $\boldsymbol{\pi}_{od}$, i.e., if the service time at the station is directly affected by the charge amount distribution. Hence, this is only an appropriate assumption when users park and charge their vehicle at the station while performing an activity such as shopping or dining. However, if this is not the case, and charging is the sole purpose of the users for coming to the station (as is the case in many fast charging stations), this uniqueness result is not valid, and there could be multiple equilibria. Further discussion on this is provided in Appendix \ref{subsection: queues} (available as supplementary material).

\section{Socially Optimal Pattern}\label{sec.so}

In the social optimal (SO) view, the CNO  has the objective to minimize the total latency and electricity cost due to the actions of all EV users throughout the entire network. Therefore, the socially optimal flow is the solution of:
	\begin{align}\label{Soptocialy}
	 &\underset{f}{\text{minimize}}
	 \sum_{a \in A} x_a t_a + \sum_{j \in J} \lambda_j T_{j}(\lambda_j)  +  
     \alpha \sum_{j \in J} D_j \bigg( U_j(Q)\bigg), 
	\\&\nonumber  \text{subject to} ~ \eqref{eq:UE.cons1}-\eqref{eq:UE.cons5}, 
	 \end{align} where $D(\cdot)$ is a strictly convex and continuously differentiable electricity cost function for each charging station that calculates the payment the CNO owes the electricity retailer. We note that for station $j$, $D_j(\cdot)$ is a function of  the expected total energy consumption $U_j(Q)$ of that charging station:
	 \begin{align} 
   U_j(Q) = \sum_{(o,d) \in \mathcal O} q_{od} \sum_{i = 1}^{K_{od}} \delta_{ij} \bigg[ E(\frac{Q_{od}^{i}}{q_{od}}) - E(\frac{Q_{od}^{i-1}}{q_{od}}) \bigg].
\end{align}

As the only congestible resources in our model are the charging stations, and each user only visits a single station on their path, the CNO's goal is to  calculate the plug-in fees and electricity prices for each charging station that drive the equilibrium flow pattern to be equal to the SO pattern. 

\begin{proposition}
\label{prop:socialoptfees}
If the condition set forth  in Theorem \ref {uniq.th.eq.} is satisfied (i.e., if the threshold-based equilibrium is unique),  there exists an anonymous charging station pricing scheme that can enforce the system optimum flow. Specifically, the socially optimal plug-in fee ($\tau_j$) and the cost of electricity ($\upsilon_j$) at charging station $j$ are as follows:
\begin{align}
    & \tau_j = \lambda_j T_{j}^{'}(\lambda_j),
     ~~\upsilon_j =  \frac{ \partial D_j \big( U_j(Q)\big)}{\partial U_j(Q)}. \label{toll.final}
\end{align}
\end{proposition} 

\noindent\textit{Proof.} The proof is in the Appendix.

We note that it would still be possible to decide station prices that minimize the social cost of the worst-case equilibrium in the case of non-unique equilibria, though this is clearly a harder problem.

\section{Numerical Experiment}
\label{sec.numerical}

\subsection{Bay Area Network (1 OD Pair) \textbf{without} Initial Energy Range Limitations}
\label{sec:no range constraints}

In this section, we present simulation results for the emergent equilibrium of a transportation network equipped with fast charging stations. We first show the structure of the Wardrop equilibrium without using plug-in fees and then we show the effects of using the socially optimal plug-in fees to reduce congestion. \textcolor{black}{In this section, we only look at 1 OD pair and we do not consider the initial energy levels of the users and the resulting range limitations (we present results from these cases in Sections \ref{sec:yes range constraints} and \ref{sec:larger example}).}

\begin{figure}
\centering
\includegraphics[width=0.55\columnwidth]{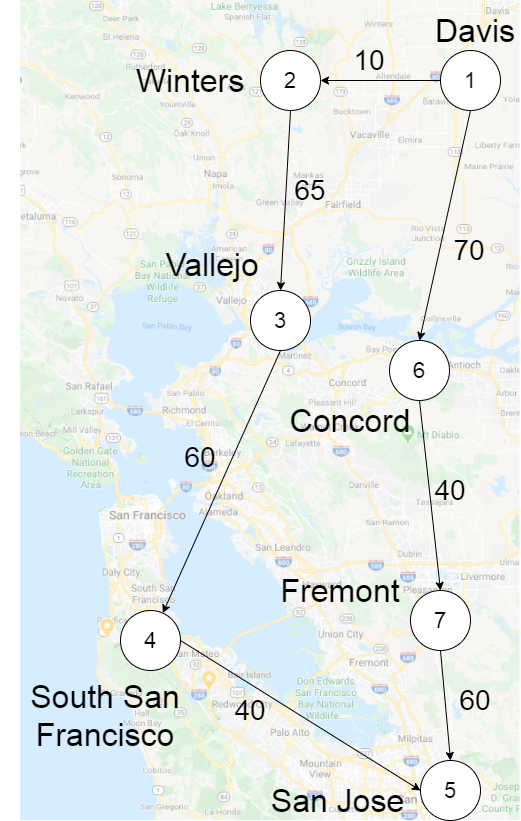} 
\caption{Bay area map with road travel times (minutes).}
\label{fig:map}
\end{figure}

In the following, we make use of a transportation network (7 arcs, 7 nodes) modeled after the San Francisco Bay Area region in Northern California as shown in Figure \ref{fig:map}. In this example we assume $q_{od} =100$ EVs/hour traveling south from Davis to San Jose with transit times (in minutes) shown in Figure \ref{fig:map}. Each of the 7 cities has fast charging available for the EVs with a maximum capacity of 10 EVs and all are assumed to have the same charging rate \textcolor{black}{(50kW DC, CHAdeMO standard)}. 
With each vehicle requiring a set of arcs connecting Davis to San Jose that enters 1 charging station, there are 9 feasible travel paths within the entire network. Furthermore, we assume that the electricity prices the CNO pays the electricity retailer are the day-ahead locational marginal prices (LMPs) for each of the locations: Davis, Winters, Vallejo, South San Francisco, San Jose, Concord, and Fremont. The electricity prices are: $[17.14 \; 17.34 \; 21.56 \; 22.25 \; 22.56 \; 21.90 \; 22.27]$ \$/MWH, respectively \cite{LMPs}. The searching time function  $T_j(\lambda_j)$ to capture congestion due to finding a charger within each city was selected as: $T_j(\lambda_j) = 0.4(\frac{\lambda_j}{x_j})^3$ (where $x_j$ is the capacity of the location).  
\textcolor{black}{
We chose the searching time function as $T_j(\lambda_j) = 0.4(\frac{\lambda_j}{x_j})^3$ because it outputs a cost close to 0 while the station still has available chargers, and then ramps up quickly once the arrival rate approaches the maximum capacity of the charging station ($x_j$), thus producing the effect that the users are spending more time trying to find an available charger. We note that this is not the only searching time function that can be used. This function can be tuned to change the effects of users waiting for available chargers at each location.
}
Last, we assume each EV user's charging demand takes a value from a uniform distribution from 0 to 80 KWH ($U[0,80]$). Note that all of these numbers are chosen for simplicity so that we can get results that can be easily parsed visually. \textcolor{black}{Last, we note that our case study formulation is convex and we made use of Matlab and CVX to solve for the equilibrium patterns. Each simulation took less than 0.5 seconds on a desktop computer with an Intel i7 processor and 16gb of RAM.} 

\textit{1) No plug-in fees:} We first show the traffic and charge patterns of the network without any plug-in fees to mitigate congestion (i.e., in this case we assume that the plug-in fees are equal to zero and the EV users pay the same grid price that the CNO pays for electricity). 

\begin{figure*}[t]
\centering
\includegraphics[width=\textwidth]{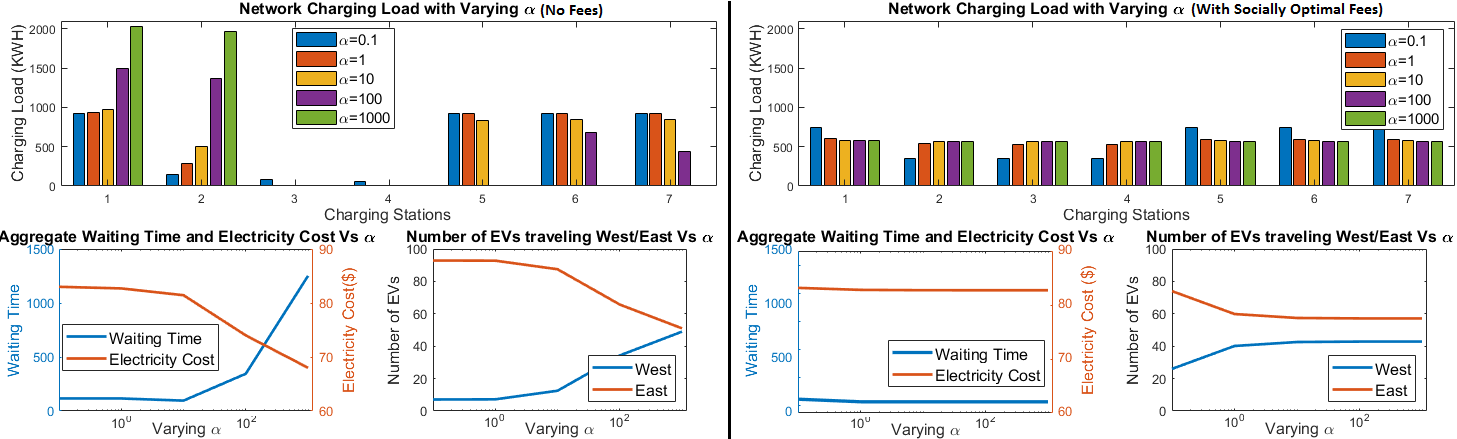} 
\caption{\textbf{Left:} (No plug-in fees) Counter-clockwise from top: 1) Expected electricity demand at the charging locations for varying values of $\alpha$, 2) Aggregate waiting time and electricity cost vs $\alpha$, 3) Traffic flow split between the east route and the west route. \textbf{Right:} Results for varying values of $\alpha$ with the socially optimal plug-in fees.}
\label{fig:wardrop_load}
\end{figure*}

In the left half of Figure \ref{fig:wardrop_load}, we present the charging load at each of the 7 locations in the network for varying values of $\alpha$. Recall, $\alpha$ is the time-value-of-money of the users in the network. When $\alpha$ is small, EV users are not strongly affected by the varying electricity prices in the network and they distribute themselves among the charging locations to minimize their waiting costs and transit times. However, as $\alpha$ increases the EV users begin to prioritize cheaper charging prices instead of shorter waiting times. In Figure \ref{fig:wardrop_load}, this is evident at locations 1 and 2 (Davis and Winters) which have the least expensive electricity prices.

At large values of $\alpha$, EV users are willing to endure longer waiting times to avoid expensive electricity costs. The bottom left plot in Figure \ref{fig:wardrop_load} shows the effects of varying $\alpha$ on the network's total electricity cost and total wait time. As shown in this plot, large values of $\alpha$ rapidly increase wait times as the population congests the cheapest locations. The plot to the right shows how the EVs split between the west route and the east route as $\alpha$ varies (note that the east route is 5 minutes quicker than the west route, but at large values of $\alpha$, the split of traffic flow is equal).

\textit{2) Socially optimal plug-in fees and electricity prices:} We now show the traffic and charge patterns of the network with the socially optimal plug-in fees as calculated in \eqref{toll.final}. \textcolor{black}{The socially optimal plug-in fees when $\alpha=1$ were: $[4.06 \; 2.90 \; 2.86 \; 2.85 \; 4.01 \; 4.02 \; 4.01]\$$ and the socially optimal plug-in fees when $\alpha=10$ were: $[3.58 \; 3.46 \; 3.42 \; 3.41 \; 3.53 \; 3.54 \; 3.54]\$$ for each of the locations: Davis, Winters, Vallejo, South San Francisco, San Jose, Concord, and Fremont, respectively. As you can see in the right half of Figure \ref{fig:wardrop_load} from the plug-in fees for the $\alpha=1$ case for locations 2, 3, and 4 (Winters $2.90\$$, Vallejo $2.86\$$, South San Francisco $2.85\$$), the plug-in fees are significantly lower than the other locations to influence more users to charge their EVs at these locations.}

As shown in the right half of Figure \ref{fig:wardrop_load}, the inclusion of socially optimal plug-in fees results in the EV charging load being split nearly equally among all charging locations for all values of $\alpha$. Furthermore, as shown in the bottom left plot, the waiting time never spikes like it did at large values of $\alpha$ without the plug-in fees. The plug-in fees are able to successfully balance minimizing electricity cost and wait times at the locations.
\begin{figure}
\centering
\includegraphics[width=0.9\columnwidth]{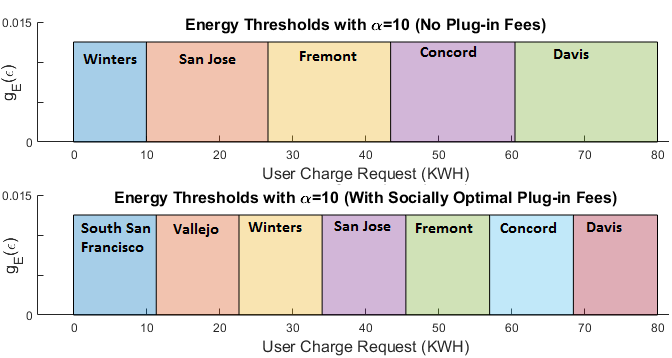} 
\caption{\textbf{Top:} Energy request partitions for EV users (no plug-in fees). \textbf{Bottom:} Energy request partitions for EV users (socially optimal plug-in fees).}
\label{fig:wardrop_energy_thresholds}
\end{figure}

In Figure \ref{fig:wardrop_energy_thresholds}, we present the energy thresholds for path selection for EV users traveling between Davis and San Jose with $\alpha=10$. The top half corresponds to the case where there are no plug-in fees and the bottom corresponds to the case with the socially-optimal plug-in fees in place. Figure \ref{fig:wardrop_energy_thresholds} shows that all the EV users with energy requests in the same shaded region will charge their EVs at the same location. Note that in the case where there are no plug-in fees (top graph), no EV will ever choose to charge in South San Francisco and Vallejo due to the 5 minute longer travel time and more expensive electricity prices. The inclusion of socially optimal plug-in fees alleviates this issue (bottom graph). Last, note that the EV users with the largest energy demands always choose to charge their EVs at the cheapest location, Davis (which matches the intuition we spoke about in Section \ref{section: network equilibrium model}).

\textcolor{black}{
\subsection{Bay Area Network (1 OD Pair) \textbf{with} Initial Energy Range Limitations}
\label{sec:yes range constraints}
}
\begin{figure}
\centering
\includegraphics[width=0.55\columnwidth]{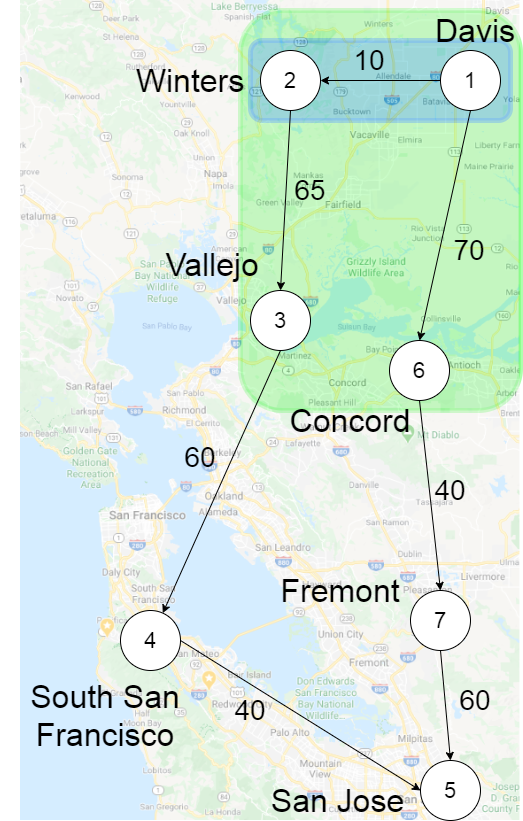} 
\caption{\textcolor{black}{Bay area map with arc travel times (minutes). The blue and green boxes represent the reachable charging stations by different groups of EVs due to their initial energy amount. Blue: EVs in this group must charge at Davis or Winters. Green: EVs in this group must charge at Davis, Winters, Vallejo, or Concord.}}
\label{fig:rangeconstrained}
\end{figure}

\textcolor{black}{
In this section, we present simulation results for the emergent equilibrium of the same Bay Area transportation network equipped with fast charging stations. However, in this case study we account for EVs with varying initial energy levels which restrict where some EVs must charge. Specifically, as shown in Figure \ref{fig:rangeconstrained}, we include 3 groups of EVs. Namely, those that have a high initial energy level (i.e., can charge at any of the 7 locations), EVs that have a medium initial energy level (i.e., can charge at any of the 4 locations in the green box), and EVs with very low initial energy (i.e., must charge at either Davis or Winters as represented by the blue box). We denote them as high initial energy, medium initial energy, and low initial energy EVs, respectively.
In the following, we simulated the population of 100 EVs traveling from Davis to San Jose with varying portions of the population split into the 3 groups to show how these users affect the resulting equilibrium with and without the socially optimal plug-in fees. The following simulations were ran with $\alpha=10$.
\begin{table*}
\begin{center}
\small
\textcolor{black}{
 \begin{tabular}{||c c c c c c c||} 
 \hline
 Case & High & Medium & Low &
 Optimal plug-in fees & Electricity Cost & Waiting Time (min)\\ [0.5ex] 
 \hline\hline
 1 & 100 & 0 & 0 & No & \$81.46 & 96.31 \\ 
 \hline
 2 & 100 & 0 & 0 & Yes & \$82.74 & 29.16 \\ 
 \hline
 3 & 0 & 0 & 100 & No & \$68.56 & 1252.10 \\ 
 \hline
 4 & 0 & 0 & 100 & Yes & \$68.56 & 1250.00 \\ 
 \hline
 5 & 0 & 100 & 0 & No & \$77.32 & 174.60 \\ 
 \hline
 6 & 0 & 100 & 0 & Yes & \$77.74 & 156.25 \\ 
 \hline
 7 & 50 & 25 & 25 & No & \$81.45 & 96.42 \\ 
 \hline
 8 & 50 & 25 & 25 & Yes & \$82.74 & 29.16 \\ 
 \hline
 9 & 25 & 25 & 50 & No & \$78.56 & 118.92 \\ 
 \hline
 10 & 25 & 25 & 50 & Yes & \$78.37 & 84.46 \\
 \hline
\end{tabular}
\caption{\textcolor{black}{Simulation results accounting for EV populations with various initial energy levels.}}
\label{table:1}
}
\end{center}
\end{table*}
}
\textcolor{black}{
In Table II, we list the number of EVs in each group, whether or not the simulation included the socially optimal plug-in fees, the total electricity cost across the network, and the total waiting time across all 7 locations. We would like to emphasize the results in the last column of Table II that lists the total waiting time in the network. In cases 1 and 2, we show the same results as in Section \ref{sec:no range constraints} where the inclusion of socially optimal plug-in fees reduces the total waiting time by 70\%. In cases 3 and 4, all 100 EVs in the population are in the low initial energy group. In this case, the Winters and Davis are highly congested and the socially optimal plug-in fees are only able to reduce the total waiting time by 0.17\%. In cases 5 and 6, the entire population is in the medium initial energy group, and the inclusion of socially optimal plug-in fees reduces the total waiting time by 10.5\%. Cases 7-10 split the population between all 3 groups. Cases 7 and 8 have a majority of the users in the high initial energy group while cases 9 and 10 have a majority of the users in the low initial energy group. The difference in total waiting time between case 7 and 8 is 70\% and the difference in total waiting time between case 9 and 10 is 29\%. As shown in these 10 cases, plug-in fees tend to be more effective (i.e., reduce total waiting time) in cases where the EV population has more charging options (i.e., plug-in fees work better when more users are able to travel to far away charging stations to alleviate congestion at nearby stations). This matches our intuition that plug-in fees will be more effective when the EV population is more flexible.
}

\subsection{Larger Bay Area Network (3 OD Pair) Example}
\label{sec:larger example}

\begin{figure*}
\centering
\includegraphics[width=1.75\columnwidth]{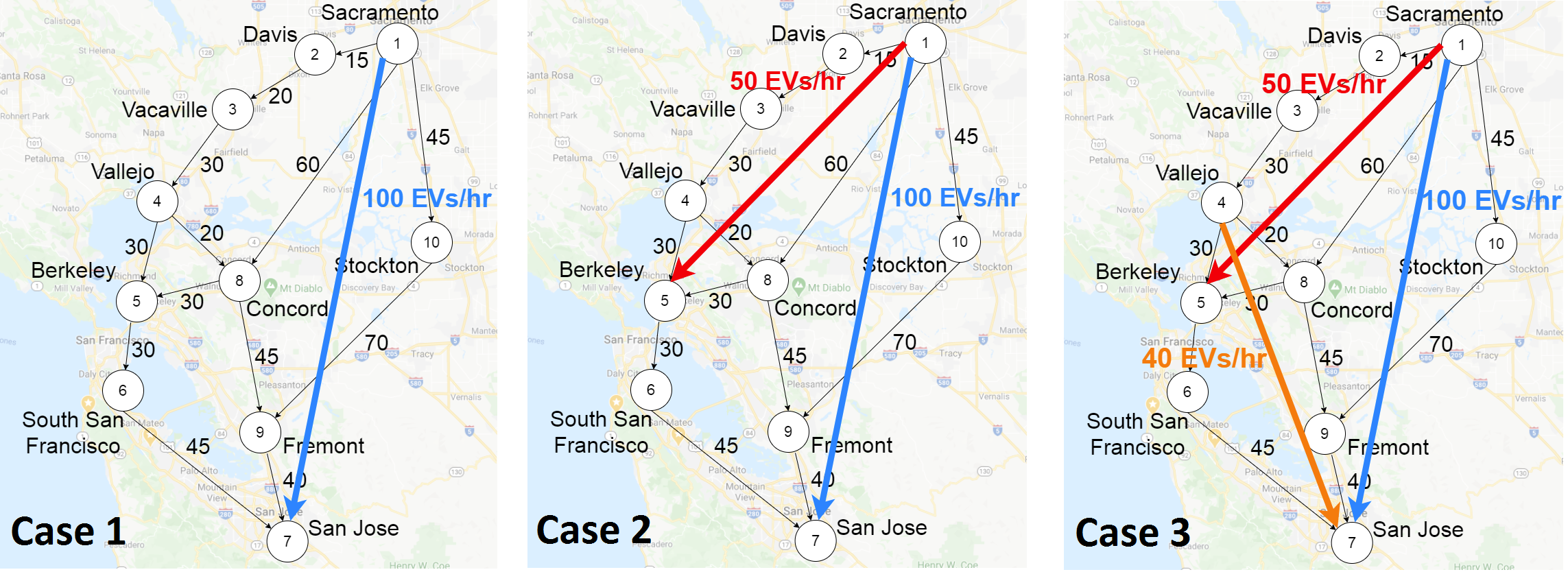} 
\caption{\textcolor{black}{Bay area map with arc travel times (minutes). Case 1 (Left): 1 Origin-Destination Pair with 100 EVs/hr traveling from Sacramento to San Jose. Case 2 (Middle): 2 Origin-Destination Pairs. Additional 50 EVs/hr traveling from Sacramento to Berkeley. Case 3 (Right): 3 Origin-Destination Pairs. Additional 40 EVs/hr traveling from Vallejo to San Jose.}}
\label{fig:bigger3casesintro}
\end{figure*}
\textcolor{black}{In this section, we present simulation results for the emergent equilibrium of a larger transportation network with multiple origin-destination pairs. As shown in Figure \ref{fig:bigger3casesintro}, we consider 3 cases. The first case considers 100 EVs/hr traveling from Sacramento to San Jose. The second case considers the 100 Sacramento-San Jose EVs/hr and adds 50 EVs/hr traveling from Sacramento to Berkeley. In Case 3, we consider the 100 Sacramento-San Jose EVs, the 50 Sacramento-Berkeley EVs, and add 40 EVs/hr traveling from Vallejo to San-Jose. \\
\indent We make use of a larger transportation network (13 arcs, 10 nodes) modeled after the San Francisco Bay Area region in Northern California as shown in Figure \ref{fig:bigger3casesintro}. Each of the 10 cities has fast charging available for the EVs with a maximum capacity of 10 EVs and all are assumed to have the same charging rate \textcolor{black}{(50kW DC, CHAdeMO standard)}. 
This network is larger than the previous, and there are 62 feasible travel paths within the entire network (instead of 9). Furthermore, we assume that the electricity prices the CNO pays the electricity retailer are the day-ahead locational marginal prices (LMPs) for each of the locations: Sacramento, Davis, Vacaville, Vallejo, Berkeley, South San Francisco, San Jose, Concord, Fremont, and Stockton. The electricity prices are: $[38.33 \; 35.80 \; 38.43 \; 39.08 \; 38.82 \; 39.53 \; 40.23 \; 38.63 \; 39.48 \; \\33.69]$ \$/MWH, respectively \cite{LMPs}. The searching time function  $T_j(\lambda_j)$ to capture congestion due to finding a charger within each city was selected as: $T_j(\lambda_j) = 0.4(\frac{\lambda_j}{x_j})^3$ (where $x_j$ is the capacity of the location) which is the same as used in the previous examples.  
Last, we assume each EV user's charging demand takes a value from a uniform distribution from 0 to 80 KWH ($U[0,80]$). 
}
\begin{figure*}
\centering
\includegraphics[width=1.75\columnwidth]{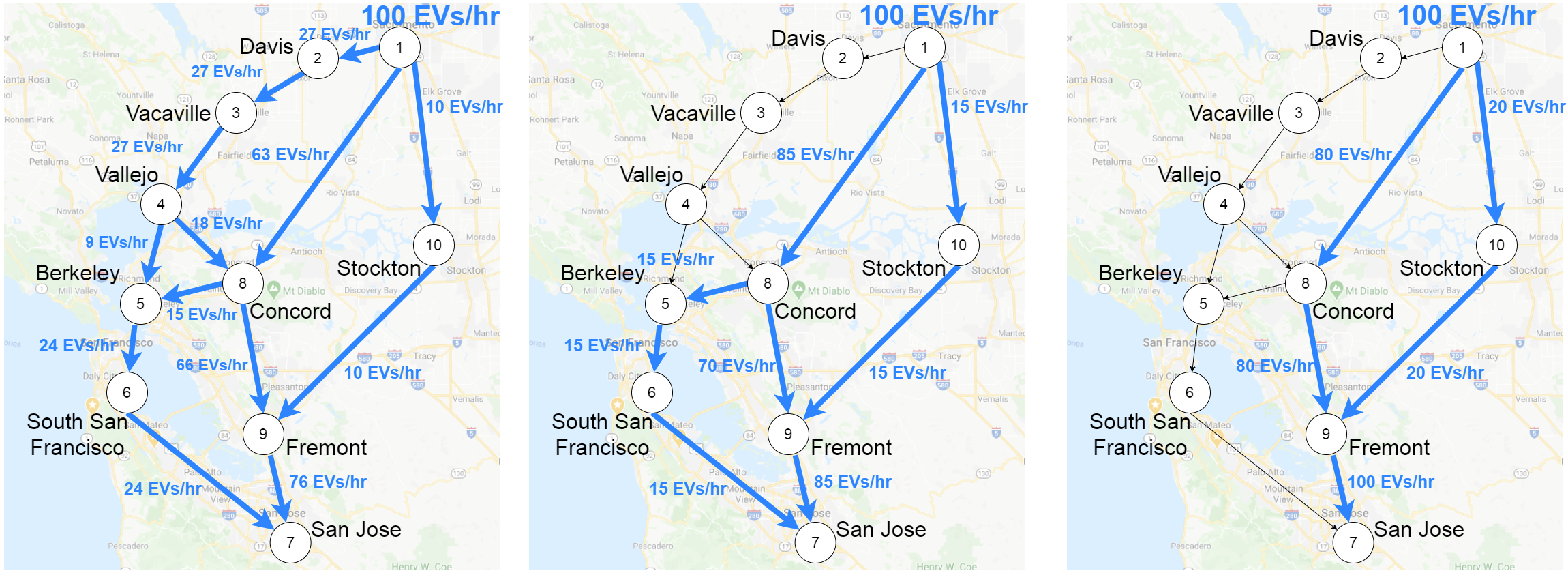} 
\caption{\textcolor{black}{Visual representation of the traffic flow of the Sacramento-San Jose OD Pair. Case 1 (Left): Sacramento-San Jose OD Pair Only, flow on all roads. Case 2 (Middle): Additional 50 EVs/hr traveling from Sacramento to Berkeley, Sacramento-San Jose travelers no longer travel through Davis, Vacaville, or Vallejo. Case 3 (Right): Additional 40 EVs/hr traveling from Vallejo to San Jose, Sacramento-San Jose travelers no longer travel through Berkeley or South San Francisco.}}
\label{fig:bigger3casesresults}
\end{figure*}

\textcolor{black}{
In Figure \ref{fig:bigger3casesresults}, we present the traffic flow of the 100 EVs in the Sacramento-San Jose OD pair for each of the 3 cases with $\alpha=25$ and with the inclusion of the socially optimal plug-in fees to mitigate congestion. As shown in Case 1 (Left), the Sacramento-San Jose EVs are able to utilize all roads and stations to travel to San Jose. In Case 2 (middle), a flow of 50 EVs/hr is added from Sacramento to Berkeley which causes congestion in Davis, Vacaville, and Vallejo. As a result, all the Sacramento-San Jose EVs travel and charge without visiting Davis, Vacaville, or Vallejo. Similarly, with the addition of 40 EVs/hr traveling from Vallejo to San Jose in Case 3 (right), there is congestion at Berkeley and South San Francisco. As a result, the Sacramento-San Jose EVs instead select routes farther east (through Concord and Stockton) to avoid congestion in South San Francisco.
}
\begin{figure}
\centering
\includegraphics[width=0.85\columnwidth]{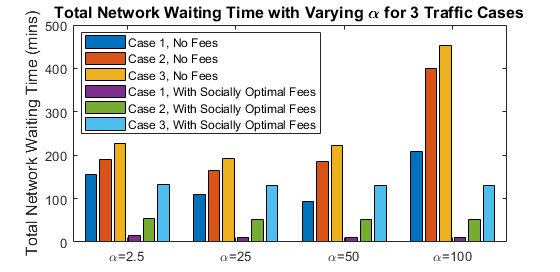}
\caption{\textcolor{black}{Total Network Waiting time at charging stations for each of the 3 cases for various values of $\alpha$ (with and without the socially optimal plug-in fees)}. }
\label{fig: mult cases waiting time results}
\end{figure}

\textcolor{black}{
Last, in Figure \ref{fig: mult cases waiting time results} we present the total network waiting time at the charging stations for various values of $\alpha$ with and without the socially optimal plug-in fees. For these results, we used EV populations from case 3, with all 3 OD pairs (190 EVs/hr traveling). As shown in Figure \ref{fig: mult cases waiting time results}, the addition of the socially optimal plug-in fees significantly reduces waiting time for all values of $\alpha$. Furthermore, the socially-optimal plug-in fees are able to induce the same minimal waiting time no matter the value of $\alpha$. From these results and the waiting time results in Section \ref{sec:yes range constraints}, it is clear that utilizing plug-in fees can significantly reduce congestion at charging stations for 1 OD pair networks and multiple OD pair networks.
}


\section{Conclusion}
\textcolor{black}{
In this paper, we studied the traffic patterns as well as the charging patterns of a population of cost-minimizing EV owners traveling and charging within a transportation network equipped with fast charging stations. We studied how the charging network operator can influence where EV users charge in order to optimize the utilization of fast charging stations. Specifically, we were able to capture the resulting equilibrium wait times and electricity load  by solving a \textit{traffic and charge assignment problem} for the fast charging station network. Our problem formulation allowed us to: 1) Study the expected station wait times as well as the probability distribution of aggregate charging load  of EVs in a FCS network in a mobility-aware fashion (an aspect unique to our work), while accounting for heterogeneities in users' travel patterns, energy demands, and geographically variant electricity prices (which would be applicable to real-world charging networks for both optimizing the usage of charging stations or for planning new stations). 2) Characterize the special threshold-based structure that determines how EV owners choose where to charge their vehicle at equilibrium, in response to the FCS's charging price structure, users' energy demands, and users' mobility patterns. As shown in Figure \ref{fig:wardrop_energy_thresholds}, we were able split up the EV population based on their energy requests and predict where each EV user will charge their vehicle. This allows the CNO to tune their plug-in fees until a desired split is achieved in the EV population. Unlike previous works, we did not need to discretize the charge requests of EVs in order to determine the equilibrium charging behavior. 3) Provide a convex optimization problem formulation to solve for the network's unique equilibrium. Furthermore, we illustrated how to induce a socially optimal charging behavior by deriving the socially optimal plug-in fees and electricity prices at the charging stations (as shown in the right half of Figure \ref{fig:wardrop_load} where the charging load is split evenly among the available stations). In future work, we hope to further examine heterogeneity in EV battery capacities that limit how far EVs can travel to charge as well as investigate the effects of home and workplace charging on the traffic and charge patterns in a FCS network.
}


\bibliographystyle{IEEEtran}
\bibliography{references}

\clearpage
\begin{center}
\textsc{Supplementary Material}

\end{center}
\textcolor{black}{
\subsection{Notation}
\begin{flalign}
    \mathcal{G}\quad &\text{Transportation network graph} \nonumber \\
    \mathcal{V}\quad &\text{Set of locations/nodes } \nonumber \\
    \mathcal{J}\quad &\text{Set of locations/nodes with chargers} \nonumber \\
    \mathcal{A}\quad &\text{Set of roads/arcs } \nonumber \\
    (o,d)\in\mathcal{O}\quad &\text{Set of origin/destination pairs } \nonumber \\
    \alpha\quad &\text{Time value of money parameter (minutes/\$) } \nonumber \\
    \epsilon\quad &\text{Energy request (KWh) } \nonumber \\
    q_{od}\quad &\text{EV arrival rate with origin/destination o,d} \nonumber \\
    \mathcal{E}\quad &\text{Energy request random variable } \nonumber \\
    g_{\mathcal{E}}^{od}\quad &\text{Energy request distribution for pair o,d } \nonumber \\
    \epsilon_{min}\quad &\text{Minimum energy request } \nonumber \\
    \epsilon_{max}\quad &\text{Maximum energy request} \nonumber \\
    \lambda_j\quad &\text{EV arrival rate at station j } \nonumber \\
    \gamma\quad &\text{Charging rate inverse (mins/kwh) } \nonumber \\
    T_j(\lambda_j)\quad &\text{Waiting time function at station j (minutes) } \nonumber \\
    p\in\mathcal{P}\quad &\text{Set of arcs/nodes of all paths } \nonumber \\
    \mathcal{P}_{od}\quad &\text{Feasible paths for pair o,d } \nonumber \\
    K_{od}\quad &|\mathcal{P}_{od}| \nonumber \\
    t_a\quad &\text{Travel time on arc/road }a \text{(minutes)}\nonumber \\
    \ell_{p_i}(\lambda)\quad &\text{Latency of path }p_i \text{ under traffic flow }\lambda \nonumber \\
    \delta\quad &\text{Indicator variable } \nonumber \\
    \tau_j\quad &\text{Plug-in fee at station j } \nonumber \\
    v_j\quad &\text{Price per kwh of electricity at station j } \nonumber \\
    f_{od}^i\quad &\text{User flow with pair o,d on path i (EVs/hr)} \nonumber \\
    x_a\quad &\text{Total flow on arc a (EVs/hr) } \nonumber \\
    \theta_{od}^i\quad &\text{Price of electricity on path i } \nonumber \\
    \pi_{od}\quad &\text{Energy threshold }\pi_{od}\in[\epsilon_{min},\epsilon_{max}] \nonumber \\
    C_{od}^{p_i}\quad &\text{Total cost of traveling on path i } \nonumber \\
    U_j(\cdot)\quad &\text{Expected total energy consumption at node j } \nonumber \\
    D_j(\cdot)\quad &\text{Cost of procuring electricity at station j } \nonumber \\
    G(\epsilon)\quad & \int_{0}^{\epsilon}g_{\mathcal{E}}(x) dx, \text{ Energy Demand CDF } \nonumber \\
    G^{-1}(t)\quad & \text{SUP}\{ \epsilon; G(\epsilon) < t\}, \text{ the inverse of CDF, } G(\epsilon) \nonumber \\
    E(x)\quad & \int_{0}^{x} G^{-1}(t) dt \nonumber \\
    Q_{od}^{i}\quad & \sum_{j=1}^{i} f_{od}^{j}, \text{ sum of flows on paths } 1,\dots,i \nonumber
\end{flalign}
}

\subsection{Proof of Lemma \ref{lemma: firstprinc}}
\begin{proof}
For any user traveling on path $i$, we know that at equilibrium, we must have $ \alpha \theta_{od}^{i} \epsilon + l_{p_i}(\lambda) + \alpha \sum_{j \in J}\delta_{ji}\tau_{j} \leq \alpha \theta_{od}^{k} \epsilon + l_{p_k}(\lambda) + \alpha \sum_{j \in J}\delta_{jk}\tau_{j}$. On the other hand, for any user traveling on path $k$, $\alpha \theta_{od}^{i} \epsilon' + l_{p_i}(\lambda) + \alpha \sum_{j \in J}\delta_{jk}\tau_{j} \geq \alpha \theta_{od}^{k} \epsilon' + l_{p_k}(\lambda) + \alpha \sum_{j \in J}\delta_{jk}\tau_{j}$. Accordingly, as we know that $ \theta_{od}^{i} \geq \theta_{od}^{k}$, we must have $\epsilon - \epsilon' 
\leq 0$.
\end{proof}

\subsection{Proof of Theorem \ref{theorem: energy thresholds}}

\begin{proof}
First, we prove that if the specific threshold-based  traffic pattern specified by \eqref{thresholds} is established, we have equilibrium. We prove this by contradiction. { Suppose the given flow pattern is not an equilibrium}. This means that there will exist a user with energy demand $\epsilon'$   assigned to path $p_{i+1}$ ($\pi_{od}^{i} \leq \epsilon'$ ) who has incentive to change his or her path to $p_k$ (without loss of generality we assume $k < i+1$). We use induction on $k$ to show  contradiction.
For the user to prefer path $k=i$, we must have ${C}_{od}^{p_{i}}(\epsilon') < {C}_{od}^{p_{i+1}}(\epsilon')$, i.e.,
\begin{align}
& \psi_{od}^{i} + \alpha \theta_{od}^{i}\epsilon' < \psi_{od}^{i+1} + \alpha \theta_{od}^{i+1}\epsilon' ~~\Rightarrow \nonumber\\&\epsilon' < \frac{\psi_{od}^{i+1} - \psi_{od}^{i} }{\alpha (\theta_{od}^{i} - \theta_{od}^{i+1})}  = \pi_{od}^{i}\end{align} which is contradictory to  $\pi_{od}^{i} \leq \epsilon'$.
Now, let us assume that the claim holds for all paths  $p_{k+1},p_{k+2},\ldots,p_i$, which  mean that we must have
 \begin{align}\label{i+1toj+1}
 & \frac{\psi_{od}^{i+1} - \psi_{od}^{k+1} }{\alpha (\theta_{od}^{k+1} - \theta_{od}^{i+1})} < \frac{\psi_{od}^{i+1} - \psi_{od}^{i} }{\alpha (\theta_{od}^{i} - \theta_{od}^{i+1})} ~~\Rightarrow \\& 
\alpha \theta_{od}^{i}(\psi_{od}^{i+1} - \psi_{od}^{k+1}) <  \bigg( \alpha \theta_{od}^{k+1}(\psi_{od}^{i+1} - \psi_{od}^{i}) - \nonumber \\& ~~~~~~~~~~~~~~~~~~~~~~~~~~~~~ \alpha \theta_{od}^{i+1}(\psi_{od}^{k+1} - \psi_{od}^{i})\bigg).
\end{align}
This means that for used path $p_k$, we should have:
\begin{align}\label{pathpk}
  \pi_{od}^{i} \leq \epsilon' <  \frac{\psi_{od}^{i+1} - \psi_{od}^{k} }{\alpha (\theta_{od}^{k} - \theta_{od}^{i+1})}.
\end{align} However, given that $\pi_{od}^{k} \leq \pi_{od}^{i}$ by definition, we have:
\begin{align}
    &\frac{\psi_{od}^{k+1} - \psi_{od}^{k} }{\alpha (\theta_{od}^{k} - \theta_{od}^{k+1})} < \frac{\psi_{od}^{i+1} - \psi_{od}^{i} }{\alpha (\theta_{od}^{i} - \theta_{od}^{i+1})} ~~\Rightarrow \nonumber \\&  
    \alpha \theta_{od}^{i}(\psi_{od}^{k+1} - \psi_{od}^{k}) - \alpha \theta_{od}^{i+1}(\psi_{od}^{i+1} - \psi_{od}^{k}) - \nonumber\\& \alpha \theta_{od}^{i+1}(\psi_{od}^{k+1} - \psi_{od}^{i}) + \alpha \theta_{od}^{k+1}(\psi_{od}^{i+1} - \psi_{od}^{i}) < \nonumber\\& \alpha (\theta_{od}^{k} - \theta_{od}^{i+1}) (\psi_{od}^{i+1} - \psi_{od}^{i}).
\end{align} If we use \eqref{i+1toj+1}, we can write: \begin{align}
    \frac{\psi_{od}^{i+1} - \psi_{od}^{k} }{\alpha (\theta_{od}^{k} - \theta_{od}^{i+1})} < \frac{\psi_{od}^{i+1} - \psi_{od}^{i} }{\alpha (\theta_{od}^{i} - \theta_{od}^{i+1})} = \pi_{od}^{i}
\end{align} which is contradictory to \eqref{pathpk}. Hence  no user will have incentive to change his or her   path  if the specific threshold-based traffic pattern by \eqref{thresholds} is established.  

  To prove the reverse, we can see that it easily follows from Lemma \ref{lemma: firstprinc} that only threshold based flows can exist at equilibrium. Hence, we only need to prove that no other threshold based flow except for that described by \eqref{thresholds} can be an equilibrium pattern. We prove this by contradiction. Suppose, there exists  another set of thresholds $(\beta_{od}^{0},\dots,\beta_{od}^{i},\dots,\beta_{od}^{K})$ which is an equilibrium pattern. Therefore, users with energy demand $\epsilon \leq \beta_{od}^{i}$ will choose path $p_i$, and no user has any incentive to change their path to path $p_{i+1}$. Hence, $\psi_{od}^{i} + \alpha \theta_{od}^i { \epsilon} \leq \psi_{od}^{i+1} + \alpha \theta_{od}^{i+1}  \epsilon$ that leads to $\epsilon \leq \pi_{od}^{i}$ that means $\beta_{od}^{i} = \pi_{od}^{i}$ which is contradictory to existence of another set of thresholds that can be equilibrium pattern.\end{proof}

\subsection{Proof of Theorem \ref{thr.UE}}

\begin{proof} We first state a lemma that slightly reformulates the equilibrium conditions. \begin{lemma}\label{Lemma.UE}
    The flow $f$ is a user equilibrium traffic pattern if and only if, it satisfies the following conditions:
\begin{enumerate}
    \item $\forall (o,d), j: f_{od}^{j} \geq 0,$
    \item for every OD pair, for every utilized path i,l:
        \begin{align} \label{Ue.COND2}
         & \psi_{od}^{i} + \sum_{k = i}^{K_{od} -1}  G^{-1}(\frac{Q_{od}^{k}}{q_{od}}) \alpha \bigg(  \theta_{od}^{k} - \theta_{od}^{k+1} \bigg)  \nonumber \\& = \psi_{od}^{l} +  \sum_{k = l}^{K_{od} -1}  G^{-1}(\frac{Q_{od}^{k}}{q_{od}}) \alpha \bigg(  \theta_{od}^{k} - \theta_{od}^{k+1} \bigg).
        \end{align} 
\end{enumerate} 
\end{lemma}

\begin{proof}
Proof of the first condition is trivial.  For the second condition, we know from \eqref{thresholds} that 
\begin{align}
    \frac{{ \psi}_{od}^{i+1} - { \psi}_{od}^{i}}{ \alpha \bigg(  \theta_{od}^{i} - \theta_{od}^{i+1} \bigg)} =  G^{-1}(\frac{Q_{od}^{i}}{q_{od}} ).
\end{align} Therefore, for every utilized path $i$ in pair (o,d): 
\begin{align}
    & \Psi_{od}^{i} +  \sum_{k = i}^{K_{od} -1}  G^{-1}(\frac{Q_{od}^{k}}{q_{od}}) \alpha \bigg(  \theta_{od}^{k} - \theta_{od}^{k+1} \bigg) = \Psi_{od}^{K_{od}}, 
\end{align} which proves the second condition.
\end{proof}
Let us now write the Lagrangian of  \eqref{eq:UE.problem} as follows:
  \begin{equation}\label{Lagranigan}
  \begin{split}
      & L_{eq}(f_{od}^{i}, Q_{od}^{i}, q_{od}, v_{od}^{i}, z_{od}) = \\& \sum_{a \in A} x_a t_a + \alpha \sum_{j \in J} \lambda_j \tau_j + \sum_{j \in J} \int_{0}^{\lambda_j} T_{j}(x) dx  
	 + \\& \alpha \sum_{(o,d) \in \mathcal O}q_{od} \left( \sum_{i = 1}^{K_{od}} \theta_{od}^{i} \left[ E(\frac{Q_{od}^{i}}{q_{od}}) - E(\frac{Q_{od}^{i-1}}{q_{od}}) \right] \right) + \\& \sum_{(o,d) \in \mathcal O}\sum_{i = 1}^{K_{od}} v_{od}^{i} \left( Q_{od}^{i} - \sum_{k\leq i} f_{od}^{k} \right) \!+\!\! \sum_{(o,d) \in \mathcal O}\!\! z_{od} \left( q_{od} - \sum_{i = 1}^{K_{od}} f_{od}^{i}   \right) 
	\end{split}
  \end{equation} Partial derivative of \eqref{Lagranigan} with respect to $f_{od}^{i}$ and $Q_{od}^{i}$ are as follows: 
  \begin{align}
      &\frac{\partial L_{eq}}{ \partial f_{od}^{i}} = \Psi_{od}^{i} - \sum_{ k = i}^{K_{od} -1} v_{od}^{k} - z_{od} \label{dr.L.f} \\&
      \frac{\partial L_{eq}}{\partial Q_{od}^{i}} =  \alpha \bigg(  \theta_{od}^{i} - \theta_{od}^{i+1} \bigg) G^{-1}(\frac{Q_{od}^{i}}{q_{od}}) + v_{od}^{i} \label{dr.L.Q} 
  \end{align}
  In equilibrium, for every two utilized path $i,l$ we must have:  \begin{align}
  &\frac{\partial L_{eq}}{\partial Q_{od}^{i}} =  \frac{\partial L_{eq}}{\partial Q_{od}^{l}} = 0,\\ &\frac{\partial L_{eq}}{ \partial f_{od}^{i}} = \frac{\partial L_{eq}}{ \partial f_{od}^{l}} = 0,  \end{align} which are equivalent to the conditions put forth in Lemma \eqref{Lemma.UE}, proving that the solution of \eqref{eq:UE.problem} is the user equilibrium traffic pattern of network.  
\end{proof}

\subsection{Proof of Theorem \ref{uniq.th.eq.}}
\begin{proof}
 Due to the strict convexity assumption on $T_j(\lambda_j)$, the term $ \sum_{a \in A} x_a t_a + \alpha \sum_{j \in J} \lambda_j \tau_j + \sum_{j \in J} \int_{0}^{\lambda_j} T_{j}(x) dx$ is continuously differentiable and strictly convex.
As $ \alpha \sum_{(o,d) \in \mathcal O}  q_{od} \left( \sum_{k = 1}^{K_{od}} \theta_{od}^{k} \left[ E(\frac{Q_{od}^{k}}{q_{od}}) - E(\frac{Q_{od}^{k-1}}{q_{od}}) \right] \right)$  is continuously differentiable, the strict convexity of the objective function hinges on  the convexity of the following term: 
\begin{align}\label{conv.UE3}
    \sum_{(o,d) \in \mathcal O} \left( q_{od} \sum_{k = 1}^{K_{od}} \theta_{od}^{k} \left[ E(\frac{Q_{od}^{k}}{q_{od}}) - E(\frac{Q_{od}^{k-1}}{q_{od}}) \right] \right),
\end{align} where we define $Q_{od}^{0} = 0$ and $Q_{od}^{K_{od}} = 1$. We prove the convexity of \eqref{conv.UE3} over each single O-D pair, i.e., we consider the convexity of functions of the following form:
\begin{align}
      a_{od}(f) = & q \sum_{k = 1}^{K} \theta^{k} \left[ E(\frac{Q^{k}}{q}) - E(\frac{Q^{k-1}}{q}) \right]. \label{Conv.UE.2}
\end{align}
We can calculate the first order derivative of \eqref{Conv.UE.2} with respect to $f^{i}$ as follows:
\begin{align}
    \frac{\partial a_{od}(f)}{\partial f^{i}} & = \sum_{k = 1}^{K-1} \frac{\partial a_{od}(f)}{\partial Q^{k}}\cdot \frac{\partial Q^{k}}{\partial f^{i}}  \\& =  
       \sum_{k = 1}^{K-1} (\theta^{k} -  \theta^{k+1}) G^{-1}(\frac{Q^{k}}{q})\Gamma_{ki}  ,\label{firs.ord.dr.f^{i}}
   \end{align} where  $\frac{\partial Q^{k}}{\partial f^{i}} = \Gamma_{ki}$, with $\Gamma_{ki}$ given by:
   \[
    \Gamma_{ki} = \left\{\begin{array}{lr}
        1 &  i \leq k,\\
        
        0 & \mbox{else}.
        \end{array}\right.
  \] Then, we calculate the second order derivative of \eqref{Conv.UE.2}:
  \begin{align}\label{sec.ord.dr.f^{i}}
      \frac{\partial ^2}{\partial f^{i} \partial f^{l}} a_{od}(f) = \sum_{k = 1}^{K-1} \bigg[ & (\theta^{k} - \theta^{k+1})  (\frac{y_k (f)}{q})\Gamma_{ki} \Gamma_{kl}\bigg],
  \end{align}
  where $
      y_{i} (f) = \frac{\partial G^{-1} (Q^{i} /q)}{\partial Q^i}$.
  For proving convexity, we will show the Hessian of \eqref{Conv.UE.2} is non-negative with respect to $f$, i.e., we consider the non-negativity of:
  \begin{align}
    & \sum_{i,l}  \frac{\partial ^2}{\partial f^{i} f^{l}} a_{od}(f) h^{i} h^{l} = \sum_{i,l}\sum_{k = 1}^{K-1}  \bigg[  (\theta^{k} - \theta^{k+1}) \nonumber\\&  (\frac{y_k (f)}{q})   \Gamma_{ki} \Gamma_{kl}h^{i} h^{l} \bigg] = \nonumber\\& \sum_{k = 1}^{K-1} \bigg[  (\theta^{k} - \theta^{k+1})  (\frac{y_k (f)}{q}) \bigg( \sum_{i,l} \Gamma_{ki}  \Gamma_{kl}h^{i} h^{l}\bigg)\bigg] = \nonumber \\& \sum_{k = 1}^{K-1} \bigg[  (\theta^{k} - \theta^{k+1})  (\frac{y_k (f)}{q})\bigg(\sum_{i}\Gamma_{kl}h^{i}\bigg)^2\bigg].
  \end{align} The coefficients $(\theta^{k} - \theta^{k+1})  (y_k (f)/q)$ are non-negative, since $y_{i} (f) = \frac{\partial G^{-1} (Q^{i} /q)}{\partial Q^i}$ is the derivative of the inverse of an increasing function (the cumulative density function) and prices are in decreasing order ($\theta^{k} \geq \theta^{k+1}$), therefore, those coefficients are non-decreasing. The non-negativeness of Hessian of \eqref{Conv.UE.2} with respect to $f^{i}$ proves the convexity of \eqref{Conv.UE.2}. Summing over all O-D pairs will prove the convexity of \eqref{conv.UE3}. Therefore, the optimization problem \eqref{eq:UE.problem} is strictly convex, and hence has a unique solution. Thus, the equilibrium pattern as well as energy thresholds are unique. 
\end{proof}

\subsection{Modeling Charging Stations as Queues}
\label{subsection: queues}
Here we discuss the implications of various modeling choices for the wait time $T_j(\lambda_j)$ at each charging station. The expected wait time can naturally be mathematically modeled based on queueing theory. For example, in the simplest case, the expected time spent $T_j(\lambda_j)$ waiting to enter charging station $j$ could be set equal to:
\begin{equation}\label{queue}
    T_j(\lambda_j) = \frac{\lambda_j/c\mu}{1-(\lambda_j/c\mu)}\cdot \frac{1}{c\mu},  
\end{equation}
where $\mu$ is the expected service time of each user in the system and $c$ is the capacity of the charging station. However, the key challenge with queueing models stems from how to define $\mu$.  We saw that in order for our uniqueness result in Theorem \ref{uniq.th.eq.}   to hold, we need $T_j(\lambda_j)$ to be a uni-variate, strictly increasing, convex function of $\lambda_j$. In this case, we would have to assume that the expected service time, $\mu$, at the station is a constant. This means that the expected amount of energy required by the users at each station, which is an endogenous parameter in our model and is a function of users' charging decisions, should not have a direct effect on their average service times in the queueing system. We noted that while this is an appropriate assumption when users park and charge their vehicle while performing an activity such as shopping or dining, in many cases this is not  the case, and charging is the sole purpose of the users for coming to the station. Here, we would need to make $\mu$ a function of the quantiles $Q^{i}_{od}$ in order to appropriately model the mapping between the distribution of service times, which varies for each specific station, and its average wait time. We discuss  the implications of adopting a model that can capture this important connection in this section.


Let us assume that $1/\mu$ corresponds to the average charge amount (given the constant rate of charge),  and $c$ denotes the station's capacity, e.g., number of chargers times their rate of charge. The quantity $\rho = \lambda/\mu$  is the so called utilization ratio of a charging station, and is essentially equivalent to the expected charging demand of the station. 
Hence, as a first step, we will derive an analytical mapping for the expected electricity consumption of charging stations (a surrogate for $\rho$ here), as a function of the quantile variables $Q^{i}_{od}$.

\begin{lemma}\label{Lemma.elec.cost}
     The expected electricity consumption at  charging station $j$, as a function of the quantile variables $Q^{i}_{od}$ and denoted as $U_j(Q)$, is computed as follows: 
     \begin{align}
     \label{averageload}
    U_j(Q) = \sum_{(o,d) \in \mathcal O} q_{od} \sum_{i=1}^{K_{od}} \delta_{ji}\bigg [E(\frac{Q_{od}^{i}}{q_{od}}) - E(\frac{Q_{od}^{i-1}}{q_{od}}) \bigg].
     \end{align}
\end{lemma}

\begin{proof} Note that since the average demand $q_{od}$ for each type of EV user is the mean of a Poisson process, it follows from Theorem \ref{theorem: energy thresholds} that the mean arrival rate at any charging station $j \in J$ also forms a Poisson process with rate $\sum_{(o,d) \in \mathcal O} q_{od}  \sum_{i = 1}^{K_{od}} \delta_{ij} \mathbb{P}[G^{-1}(\frac{Q_{od}^{i-1}}{q_{od}})\leq x\leq G^{-1}(\frac{Q_{od}^{i}}{q_{od}})]$ \footnote{ 
When the system is at equilibrium, we are given the thresholds $\pi_{od}^{i}$, and the flow of users from  each OD pair that will choose charging station $i$ (located on path $p_i$) is determined by sampling the arrival Poisson process with constant probability equal to $\int_{\pi_{od}^{i-1}}^{\pi_{od}^{i}} g_{\randE}(\epsilon) d\epsilon $. The sampling of a Poisson process will lead to another Poisson process, which determines a part of the traffic to station $i$. The arrival process of the charging station is then the sum of  independent arrival processes from different OD pairs (conditioned on the thresholds, the users path decisions are independent at equilibrium),  and the sum of independent Poisson processes is also a Poisson process.
}.
Hence, we can write that the average electricity consumption in each charging station $j$ by using the linearity  of expectation and using Wald's equation as follows:
\begin{align}
   & \sum_{(o,d) \in \mathcal O} q_{od}  \bigg[\sum_{i = 1}^{K_{od}} \delta_{ij} \mathbb{P}[G^{-1}(\frac{Q_{od}^{i-1}}{q_{od}})\leq x\leq G^{-1}(\frac{Q_{od}^{i}}{q_{od}})] \label{} \nonumber\\&~~~~~~~\bigg(\int x g\big(x | G^{-1}(\frac{Q_{od}^{i-1}}{q_{od}})\leq x \leq G^{-1}(\frac{Q_{od}^{i}}{q_{od}})\big) dx \bigg) \bigg] \\& \quad   = \sum_{(o,d) \in \mathcal O} q_{od} \sum_{i =1}^{K_{od}} \delta_{ij} \int_{G^{-1}(\frac{Q_{od}^{^{i-1}}}{q_{od}})}^{G^{-1}(\frac{Q_{od}^{^{i}}}{q_{od}})} x g(x) dx, 
\end{align} Then we use the substitutions $l = G(x)$ or $x = G^{-1}(l)$, we will have: 
\begin{align}
   \sum_{(o,d) \in \mathcal O} q_{od} \sum_{i =1}^{K_{od}} \delta_{ij} \int_{\frac{Q_{od}^{i-1}}{q_{od}}}^{\frac{Q_{od}^{i}}{q_{od}}} G^{-1}(l) G^{-1'}(l) g(G^{-1}(l)) dl,
\end{align}and using the derivative of inverse function rule, this yields:
\begin{align}\label{conv.So.2}
    \sum_{(o,d) \in \mathcal O} q_{od} \sum_{i = 1}^{K_{od}} \delta_{ij} \int_{\frac{Q_{od}^{i-1}}{q_{od}}}^{\frac{Q_{od}^{i}}{q_{od}}} G^{-1} (l) dl.
\end{align} Therefore, the expected total electricity consumption in each charging station $j$, denoted by $U_j(Q)$, is given by: \begin{align} \label{averageload2}
   U_j(Q) = \sum_{(o,d) \in \mathcal O} q_{od} \sum_{i = 1}^{K_{od}} \delta_{ij} \bigg[ E(\frac{Q_{od}^{i}}{q_{od}}) - E(\frac{Q_{od}^{i-1}}{q_{od}}) \bigg].
\end{align}\end{proof}

In Lemma \ref{Lemma.elec.cost}, the term $\delta_{ji}[E(\frac{Q_{od}^{i}}{q_{od}}) - E(\frac{Q_{od}^{i-1}}{q_{od}})]$ corresponds to the expected energy request of an EV user in the $i$th partition for o-d pair $(o,d)$ that chooses to charge at station $j$. Summing over all o-d pairs and and all feasible paths, we get the total expected energy consumption at station $j$.

Equipped with this, we can now create a wait time resembling  the form presented in \eqref{queue} by modifying the wait times associated with finding a free EVSE at a charging station in \eqref{eq:costfunction2} with a function of the following form:
\begin{equation}\label{newwait}
    \mbox{expected wait time} = \frac{1}{c}\cdot T_j\left(\frac{1}{c}U_j(Q)\right)\epsilon,
\end{equation}
where $T_j(x) = x/(1-x)$ corresponds to the term $\frac{\lambda/c\mu}{1-(\lambda/c\mu)}$ in \eqref{queue}, and  $\epsilon$ replaces $1/\mu$ for each user. 

While we remove the details for brevity, we can show that the structure of the Wardrop equilibrium of the network given this new wait time formula is still threshold based, and can be characterized by replacing the term $\int_{0}^{\lambda_j} T_{j}(x)dx $ in \eqref{eq:UE.problem} with $\int_{0}^{\frac{1}{c}U_j(Q)} T_{j}(x) dx$. This is straightforward to show given the following lemma.
\begin{lemma}
     The derivative of $U_j(Q)$ with respect to the upper and lower quantiles $Q_{od}^{i-1}$ and $Q_{od}^{i}$ of path $i$ is respectively equal to $G^{-1}(\frac{Q_{od}^{i-1}}{q_{od}})$ and $G^{-1}(\frac{Q_{od}^{i}}{q_{od}})$, the upper and lower energy thresholds of path $i$.
\end{lemma}

However, while this all seems positive, the function $U_j(Q)$ is generally not convex in $Q$,  leading to the following result.
\begin{proposition}
The traffic flow pattern in a threshold-based EV equilibrium may not be unique in path flows when adopting wait times of the form given in \eqref{newwait}.
\end{proposition}

\subsection{Proof of Proposition \ref{prop:socialoptfees}}
\begin{proof}
The proof is straightforward. By comparing the KKT conditions of \eqref{Soptocialy}   and \eqref{eq:UE.problem}    for used path $p_i$, we must have:
\begin{align} \label{fee.1}
   &\frac{\partial L_{eq}}{\partial f_{od}^{i}} = \frac{\partial L_{so}}{\partial f_{od}^{i}} ~ \Rightarrow ~
   \sum_{j \in J} \delta_{ji}\tau_j   =  \sum_{j \in J} \delta_{ji} \lambda_j T_{j}^{'}(\lambda_j)  \\&  \frac{\partial L_{eq}}{\partial Q_{od}^{i}} = \frac{\partial L_{so}}{\partial Q_{od}^{i}} ~ \Rightarrow ~ \nonumber\\&
   \alpha \sum_{j \in J}  \upsilon_j \frac{\partial U_j (Q)}{\partial Q_{od}^{i}} = 
   \alpha \sum_{j \in J}  \frac{ \partial D_j \bigg( U_j(Q)\bigg)}{\partial U_{j}(Q)}\bigg( \frac{\partial U_j (Q)}{\partial Q_{od}^{i}} \bigg)
\end{align}
Accordingly, to derive charging station prices that induce the system-optimal flow pattern as a UE, the CNO can simply  set  the plug-in fee ($\tau_j$) and the cost of electricity ($\upsilon_j$) at charging station $j$  as follows:
\begin{align}
    & \tau_j = \lambda_j T_{j}^{'}(\lambda_j),
     ~~\upsilon_j =  \frac{ \partial D_j \bigg( U_j(Q)\bigg)}{\partial U_j(Q)}. \label{toll.final2}
\end{align}\end{proof}

\end{document}